\documentclass[conference,final]{IEEEtran}

\usepackage[utf8]{inputenc}
\usepackage[linesnumbered,lined, algoruled]{algorithm2e}
\usepackage{algorithmic}
\usepackage{amsmath}
\usepackage{amsthm}
\usepackage{amsfonts}
\usepackage{amssymb}
\usepackage{balance}
\usepackage{bm}
\usepackage{color,soul}
\usepackage[acronym,shortcuts]{glossaries}
\usepackage{graphicx}
\usepackage{graphics}
\makeglossaries
\newacronym{auc}{AUC}{area under the curve}
\newacronym{bs}{AP}{access point}
\newacronym{ce}{CE}{cross entropy}
\newacronym{cdf}{CDF}{cumulative distribution function}
\newacronym{fa}{FA}{false alarm}
\newacronym{kl}{K-L}{Kullback-Leibler}
\newacronym{ls}{LS}{least-squares}
\newacronym{llr}{LLR}{log likelihood-ratio}
\newacronym{los}{LOS}{line of sight}
\newacronym{lssvm}{LS-SVM}{least squares SVM}
\newacronym{md}{MD}{misdetection}
\newacronym{ml}{ML}{machine learning}
\newacronym{mlp}{MLP}{multy-layer perceptron}
\newacronym{mse}{MSE}{mean squared error}
\newacronym[\glslongpluralkey={neural networks}]{nn}{NN}{neural network}
\newacronym{np}{N-P}{Neyman-Pearson}
\newacronym{oclssvm}{OCLSSVM}{one-class least-square \ac{svm}}
\newacronym{pdf}{PDF}{probability density function}
\newacronym{pmd}{PMD}{probability mass distribution}
\newacronym{pso}{PSO}{particle swarm optimization}
\newacronym{rnn}{RNN}{replicator neural network}
\newacronym{roc}{ROC}{receiver operating characteristic}
\newacronym{roi}{ROI}{region of interest}
\newacronym{rss}{RSS}{received signal strength}
\newacronym[\glslongpluralkey={support vector machines}]{svm}{SVM}{support vector machine}
\newacronym{ue}{UE}{user equipment}
\newacronym{wsn}{WSN}{wireless sensor network}
\newacronym{irlv}{IRLV}{in-region location verification}

\newcommand{\cross}[2]{H_{#1}(#2)}
\newcommand{\hatcross}[2]{\hat{H}_{#1}(#2)}
\newcommand{\gy}{g(\bm a)}
\newcommand{\E}[2]{\mathbb{E}_{#1}\left[#2\right]}
\newcommand{\pr}[1]{\mathbb{P} \left[ #1 \right]}

\newtheorem{theorem}{Theorem}

\title{Location-Verification and Network Planning \\ via Machine Learning Approaches}
\author{Alessandro Brighente, Francesco Formaggio, Marco Centenaro, Giorgio Maria Di Nunzio, and    Stefano Tomasin \\ {\small Department of Information Engineering, University of Padova, via G. Gradenigo 6/B, Padova, Italy. first.lastname@dei.unipd.it} }
\date{today}

\begin{document}

\maketitle

\begin{abstract}
\Ac{irlv} in wireless networks is the problem of deciding if  \ac{ue}  is transmitting from inside or outside a specific physical region (e.g., a safe room). The decision process exploits the features of the channel between the \ac{ue} and a set of network \acp{bs}.  We propose a solution based on   \ac{ml} implemented by a \ac{nn} trained with the channel features (in particular, noisy attenuation values) collected by the \acp{bs} for various positions both inside and outside the specific region. The output is a decision on the  \ac{ue} position (inside or outside the region). By seeing  \ac{irlv}  as an hypothesis testing problem, we address the optimal positioning of the \acp{bs} for minimizing either the \ac{auc} of the \ac{roc} or  the \ac{ce} between the \ac{nn} output and ground truth (available during the training). In order to solve the minimization problem we propose a two-stage \ac{pso} algorithm. We show that for a long training and a \ac{nn} with enough neurons the proposed solution achieves the performance of the \ac{np} lemma.
\end{abstract}

\glsresetall

\section{Introduction}

Applications using information on the user location are rapidly spreading, also to ensure that some services are obtained only in pre-determined areas. %For example, some on line gambling sites can be used only when in the state where the site is registered. 
In order to establish the user position we can rely on the device itself,  requested to report the position provided by its GPS module. However, tempering with the GPS module  or its  software interface is relatively easy \cite{ceccato2018exploiting}. Thus, more reliable solutions must be explored. Location verification systems aim at verifying the position of devices  \cite{Zeng-survey, 8376254}, possibly using distance measures  obtained, for example, through the \ac{rss} at anchor nodes for signals transmitted by the terminal under verification. This problem is closely related to  {\em user authentication} at the physical layer, where wireless channel features are exploited to verify the sender of a message \cite{7270404}.

We focus here on  \ac{irlv},  the problem of deciding whether a message coming from a terminal over a wireless network has been originated from a specific physical region, e.g., a safe room, or not \cite{Zeng-survey}. \ac{irlv} can be seen as an hypothesis testing problem between two alternatives, namely being inside or outside the specific region. Among proposed solutions, we recall distance bounding techniques with rapid exchanges of packets between the verifier and the prover \cite{Brands}, also using radio-frequency and ultrasound signals \cite{Sastry}, and solutions based on anchor nodes and increasing transmit power by the sender \cite{Vora}. More recently, a delay-based verification technique has been proposed  in \cite{7145434}, leveraging geometric properties of triangles, which prevent an adversary from manipulating measured delays. 

In this paper, we consider the \ac{irlv} problem for a \ac{ue} connected to a set of network \acp{bs}. The decision on the user position is taken on the basis of observed features of the channel over which communication occurs. For exemplary purposes we focus here on the observation of the attenuation of the channels between the \ac{ue} and the \acp{bs}. We propose a \ac{ml} approach  where i)  channel measurements are collected by trusted nodes both inside and outside the \ac{roi}, ii)  a machine is trained to take decisions between the two hypotheses, iii) the machine is exploited to take decisions on the unknown \acp{ue} in the exploitation phase. \ac{ml} techniques have already found application in user authentication (see  \cite{xiao-2018} and references therein), however never in \ac{irlv}, to the best of authors' knowledge. The \ac{nn} training is based on the \ac{ce}, and, framing  \ac{irlv} problem into an hypothesis testing problem, we establish the optimality of this criterion according to the \ac{np} lemma, in asymptotic conditions. 

Then, we address the problem of optimum positioning of the \acp{bs} (network planning) for \ac{irlv}. Two metrics are considered for this optimization: a) the \ac{ce} of \ac{nn} training, and b) the \ac{auc} of the \ac{roc} of the hypothesis test. While \ac{ce} is directly related to the \ac{nn} training, the \ac{auc} is more connected to the final performance that we expect from \ac{irlv}. A limited number of neurons, as well as a limited size of the training set may provide different results for the two metrics. For the optimization of \acp{bs} position  we propose a two-stage \ac{pso} algorithm, minimizing either the \ac{ce} or the \ac{roc} \ac{auc}. Simulation results over channels with shadow fading show the merits of the proposed solution, and its effectiveness in providing reliable \ac{irlv}.

\section{System Model}\label{sec:sys model}

We consider a cellular system with $N_{\rm AP}$ \acp{bs} covering a region $\mathcal{A}$ over a plane. We propose a \ac{irlv} system able to determine if a \ac{ue} is transmitting from inside an {\em authorized} sub-region $\mathcal{A}_0 \subset \mathcal{A}$. The dependency on location of the \ac{ue}-\acp{bs} channels is exploited to distinguish between transmissions from inside and outside $\mathcal{A}_0$. Transmissions are assumed to be narrowband and the channel feature used for \ac{irlv} is its attenuation.

The \ac{irlv} procedure comprises two phases. In the first phase, named identification or training, a trusted \ac{ue} transmits a training signal (known at the \acs{bs}) from various points inside region $\mathcal{A}_0$.  The \acp{bs} estimate the channel attenuations and store them in association with $\mathcal{A}_0$. Some external authentication technique on the transmitted packet  ensures that the received signal upon which the attenuation is estimated is actually transmitted by the trusted \ac{ue}. Similarly, attenuation values are collected when the trusted \ac{ue} transmits  from the complementary region $\mathcal{A}_1$ and  stored by the \acp{bs} in association to $\mathcal{A}_1 = \mathcal{A} \setminus \mathcal{A}_0$. In the second  phase, named verification or exploitation, the \ac{ue} transmits a known training sequence from any point in $\mathcal{A}$ and the \ac{irlv} system must decide whether the \ac{ue} is in region $\mathcal{A}_0$ or $\mathcal{A}_1$.

% The location dependency of the features of the transmission channel can be further enhanced by properly placing the \acp{bs} over the area $\mathcal{A}$. In particular, different positioning lead to different shadowing incurred by the transmission of the \acp{ue} toward the \acp{bs}. Our aim is to find the optimal \acp{bs} positioning such that the authentication system can optimally discriminate between different areas based on the estimated attenuation values.

\subsection{Channel Model}

Let $\bm{x}_{\rm ap}^{(n)} =(X_{\rm ap}^{(n)},Y_{\rm ap}^{(n)})$ be the position of the $n$-th \ac{bs}. For a \ac{ue} located at $\bm{x}_{\rm ue}=(X_u,Y_u)$, its distance from \ac{bs} $n$ is denoted as $L(\bm{x}_{\rm ue},\bm{x}_{\rm ap}^{(n)})$. We assume that  the \ac{ue} transmits with constant power so that  \ac{bs} $n$ can estimate the attenuation $a{(n)}$ incurred over the channel, including the effects of path-loss and shadowing. Let $\bm a = [a(1), a(2), \ldots, a(N_{\rm AP})]$ collect attenuation values from all \ac{bs}.

 Denoting the path-loss coefficient as $a_{\rm PL}{(n)}$  the shadowing component is log-normally distributed, i.e., $\left(a_{\rm S}{(n)}\right)_{\rm dB} \sim \mathcal{N}(0,\sigma_s^2)$, and we have
\begin{equation}
    \left(a{(n)}\right)_{\rm dB} = \left(a_{\rm PL}{(n)}\right)_{\rm dB} + \left(a_s{(n)}\right)_{\rm dB}.
\end{equation}
The channel model for path-loss and shadowing is derived from \cite{3gpp}. 
For a \ac{los} link the path loss coefficient in dB is modelled as
\begin{equation}\label{eq:los}
    a_{\rm PL-LOS}{(n)} = 20\log_{10}\left(\frac{f_0 4\pi L(\bm{x}_{\rm UE},\bm{x}_{\rm AP}^{(n)})}{c}\right),
\end{equation}
where $f_0$ is the carrier frequency, and $c$ is the speed of light. For a  non-\ac{los} link the path loss coefficient in dB is defined as
\begin{equation}
\begin{split}
    &a_{\rm PL-LOS}{(n)} = 40\left(1-4 \cdot 10^{-3} \left. h_{\rm AP} \right|_{\rm m} \right)  \times  \\
    &\log_{10}\left (\left.\frac{L(\bm{x}_{\rm UE},\bm{x}_{\rm AP}^{(n)})}{10^3}\right|_{\rm m}\right ) -18\log_{10}\left.h_{\rm AP}\right|_{\rm m} +\\
    &
    + 21\log_{10}\left(\frac{\left.f_0\right|_{\rm MHz}}{10^6}\right) + 80,
    \end{split}
\end{equation}
where $\left.h_{\rm AP}\right|_{\rm m}$ is the \ac{bs} antenna elevation in meters, $\left. f_0 \right|_{\rm MHz}$ is the carrier frequency in MHz, and $\left. L(\bm{x}_{\rm UE},\bm{x}_{\rm AP}^{(n)})\right|$ is the \ac{ue}-\ac{bs} $n$ distance in meters.
We assume that correlation between shadowing coefficients $\left( a_{\rm S}(i)\right)_{\rm dB}$ and $\left( a_{\rm S}(j)\right)_{\rm dB}$ for two \acp{bs} located in $\bm{x}_i$ and $\bm{x}_j$ when the \ac{ue} is transmitting is
\begin{equation}\label{eq: coor mat}
    \E{\bm x}{a_{\rm S}(i)a_{\rm S}(j)} = \sigma_s^2\exp \left({-\frac{L(\bm{x}_i,\bm{x}_j)}{d_c}}\right),
\end{equation}
where $d_c$ is the shadowing decorrelation distance and $\E{\bm x}{\cdot}$ is the expected value with respect to the distribution of $\bm x$.

\section{In-region Location Verification}\label{sec: ml}

Let us define the two hypotheses of the \ac{irlv} problem as
 \begin{equation}
  \mathcal H_0: \mbox{the \ac{ue} is in $\mathcal{A}_0$}, \quad  \mathcal H_1: \mbox{the \ac{ue} is in $\mathcal{A}_1$}.
\end{equation}
In the training phase the \ac{ue} transmits from $S$ locations. For transmission $i=1,\ldots,S$, let $\bm{a}^{(i)}=[a^{(i)}(1),\ldots,a^{(i)}(N_{\rm AP})]$ 
be the  vector of  measured attenuations. We associate to $\bm a ^{(i)}$ the label $t_i=0$ if \ac{ue} is transmitting from inside $\mathcal{A}_0$ (hypothesis $\mathcal{H}_0$), and $t_i=1$ otherwise. Let also define $\mathcal{T} =
 \{\bm{a}^{(1)}, \ldots , \bm{a}^{(S)} \}$.
By using these attenuation training vectors and labels, we aim at building a function
\begin{equation}
    \hat{t} = g(\bm{a}) \in \{0,1\}\,
\end{equation}
that maps any attenuation vector $\bm{a}$ into a decision on the location of the \ac{ue}. We would like to have $\hat{t}=0$ if $\bm a$ was obtained when the \ac{ue} was inside $\mathcal{A}_0$ and $\hat{t}=1$ otherwise.

The performance of the \ac{irlv} system is assessed in terms of two error probabilities: the \ac{fa} probability, i.e., the probability  that a \ac{ue} in $\mathcal A_0$ is declared outside this region, and the \ac{md} probability, i.e., the probability that a \ac{ue} outside $\mathcal A_0$ is declared inside the region. In formulas, denoting with $\pr{\cdot}$ the probability function,
\begin{equation}
P_{\rm FA} =\pr{\hat{t} =1 | \mathcal H_0}, \quad  
 P_{\rm MD}=\pr{\hat{t} = 0 | \mathcal H_1}. 
\end{equation}

\subsection{Test for Known Attenuation Statistics}

The \ac{irlv} problem can be seen as an hypothesis testing problem between the two hypotheses $\mathcal H_0$ and $\mathcal H_1$. When the statistics of the attenuation vectors are known under the two hypotheses, the most powerful test for the \ac{irlv} problem is provided by the \ac{np} lemma. In particular, let us  define the \ac{llr}
\begin{equation}\label{eq:lr}
    \mathcal{L}{(\bm a)}=\log\left(\frac{p_{\bm a}(\bm a|\mathcal{H}_0)}{p_{\bm a}(\bm a|\mathcal{H}_1)}\right)\,,
\end{equation}
where $p_{\bm a|\mathcal{H}}(\bm a|\mathcal{H}_i)$ is the \ac{pdf} of the random vector $\bm a$ modelling all attenuation values $\bm a$, given that hypothesis $\mathcal H_i$ is verified, and $\log$ denotes the base-2 logarithm. The \ac{np} test function is 
\begin{equation}
\label{eq:thrOpt}
    \hat{t} = g(\bm a) = \begin{cases}
    0 & \mathcal{L}{(\bm a)} \geq \theta\,, \\ 
    1 & \mathcal{L}{(\bm a)} < \theta\,, 
    \end{cases}
\end{equation}
where $\theta$ is a threshold to be chosen in order to ensure the desired \ac{fa} probability. This test ensures that for the given \ac{fa} probability the \ac{md} probability is minimized. 

\subsection{Example of \ac{np} Test}
\label{sec:los}
We now describe an example of application of the \ac{np} test, where we can easily obtain a close-form expression for $f(\bm a)$. 

%Fig. \ref{fig:scen} shows the scenario over which \ac{np} test is applied. 
Let us define the overall network region as a circle $\mathcal{A}_c$ with radius $R_{\rm out}$ and consider a single \ac{bs} located at the center of $\mathcal{A}_c$. Consider $\mathcal{A}_{0}$ as a rectangle with nearest point to the center of $\mathcal{A}_c$ at a distance $R_{\rm min}$. The outside region is $\mathcal{A}_1 = \mathcal{A}_c \setminus \mathcal{A}_0$. In the \ac{los} scenario the scalar attenuation $a$ incurred by a \ac{ue} is given by path loss, which only depends on its relative distance to the \ac{bs}. Considering an attenuation value $a$, the \ac{ue}-\ac{bs} distance  is given by
\begin{equation}
    R = \frac{c a}{ 4 \pi f_0}.
\end{equation}
Therefore, instead of considering $p_{a|\mathcal{H}}(a|\mathcal H_i)$ we consider $p_{R|\mathcal{H}}(r|\mathcal H_i)$, where distance $R$ corresponds to attenuation $a$. We first derive the \ac{cdf} of $R$ in $\mathcal{A}_0$, i.e.,  the probability that the \ac{ue} is located in $\mathcal{A}_0$ at a distance $R\le r$ from the \ac{bs}. This is
\begin{equation}\label{eq:cdf}
     \pr{R \le r|\mathcal{H}_0} = \frac{1}{|\mathcal{A}_0|}\int_{R_{\rm min}}^{r} \rho \alpha(\rho) d\rho,     
\end{equation}
where $\alpha(R)$ denotes the angle of the circular sector measured from a distance $R$ and intersecting region $\mathcal{A}_0$ and $|\mathcal{A}_0|$ is the area of region $\mathcal{A}_0$. Then by taking the derivative of the \ac{cdf} (\ref{eq:cdf}) with respect to $r$ we obtain the \ac{pdf} 
\begin{equation}\label{eq:num}
    p_{R|\mathcal{H}}(r|\mathcal{H}_0) = \frac{1}{|\mathcal{A}_0|}r\alpha(r).
\end{equation}
Following the same reasoning and considering that the length of the circular sector with radius $r$ located in $\mathcal{A}_1$ is $2\pi - \alpha(r)$, we obtain the \ac{pdf} of transmission from a distance $r$ in $\mathcal{A}_1$ as
\begin{equation}\label{eq:den}
     p_{R|\mathcal{H}}(r|\mathcal{H}_1) = \frac{1}{|\mathcal{A}_1|}r\left(2\pi-\alpha(r)\right).
\end{equation}
From (\ref{eq:num}) and (\ref{eq:den}) we obtain the \ac{llr} as a function of the \ac{ue}'s distance from the \ac{bs} as 
\begin{equation}
    \mathcal{L}{(a)}=\log\left[\frac{|\mathcal{A}_1|\alpha \left(\frac{c a}{f_0 4 \pi}\right)}{|\mathcal{A}_0|\left(2\pi-\alpha\left(\frac{c a}{f_0 4 \pi}\right)\right)}\right].
\end{equation}
%This result will be used in the numerical section part to compare the classifier performance of the \ac{nn} with the \ac{np} test.

% \begin{figure} 
%     \centering
%     \includegraphics[width=0.8\columnwidth]{simpleScen.eps}
%     \caption{Realization of the scenario over which \ac{np} hypothesis testing is performed: the \ac{ue} transmits message toward the \ac{bs} located at the center of the overall circular region. }
%     \label{fig:scen}
% \end{figure}

\subsection{Neural Network Implementation}\label{sec:nn}

Under more complicated scenarios, it becomes hard to obtain close-form expressions for the \ac{llr}. Therefore, we reosrt to a \ac{ml} approach, using a \ac{nn}  trained with attenuation vectors $\bm{a}^{(i)}$ and labels $t_i \in \{0,1\}$. In the verification phase the trained \ac{nn} is used on the test attenuation vectors $\bm a$ to provide the decision $\hat{t} \in \{0,1\}$. Now, $g(\cdot)$ is the function implemented by the \ac{nn}.

We now provide a short description of a \ac{nn}. A feed-forward \ac{nn} processes the input in stages, named layers, where the output of one layer is the input of the next layer. The input of the \ac{nn} is $\bm{y}^{(0)} = \bm{a}$, and layer $\ell-1$ has $N^{(\ell-1)}$ outputs obtained by processing the inputs with $N^{(\ell-1)}$ scalar functions named neurons. The output of the $n$-th neuron of the $\ell$-th layer is
\begin{equation}\label{eq:nonLin}
y_n^{(\ell)} = \sigma\left( \bm{w}_n^{(\ell -1)}\bm{y}^{(\ell-1)}+b_n^{(\ell)} \right),
\end{equation}
where $\bm{w}_n^{(\ell -1)}$ and $b_n^{(\ell)}$ are coefficients to be determined in the training phase, and $\sigma(\cdot)$ is the sigmoid activation function. 
%The neuron maps via  $\psi$ a  linear combination with weights $\bm{w}_n^{(\ell -1)}\in \mathbb{R}^{1\times N^{(\ell-1)}}$ of the outputs $\bm{y}^{(\ell-1)} \in \mathbb{R}^{N^{(\ell-1)} \times 1 }$ of the previous layer, plus a bias $b_n^{(\ell)} \in \mathbb{R}^{N^{(\ell-1)} \times 1 }$. 
The last layer comprises only one neuron, $y^{(L)}$, and the final output of the \ac{nn} is the scalar 
\begin{equation}
	\tilde{t}(\bm a) \triangleq \sigma(y^{(L)}),	
\end{equation}
where $L$ is the total number of layers. Finally, the test function is obtained by thresholding $\tilde{t}(\bm{a})$, i.e.,
\begin{equation}
\label{eq:decNN}
    g(\bm{a}) = \begin{cases}
    1 & \tilde{t}(\bm a) > \lambda \\
    0 & \tilde{t}(\bm a) \leq \lambda.
    \end{cases}
\end{equation}
By varying $\lambda$ we obtain different values of $P_{\rm FA}$ and $P_{\rm MD}$ for this \ac{irlv} test.

Various options have been proposed in the literature for \ac{nn} training. We consider here as objective function the empirical \ac{ce} between the \ac{nn} output and the labels $t_i$, defined as
\begin{equation}\label{eq:ce}
\begin{split}
\hatcross{p_{\mathcal{H}|\bm a}}{g} \triangleq& -\frac{1}{S} \sum_{i=1}^{S}\left[t_i\log \tilde{t}(\bm a^{(i)}) + \right.\\
&\left. +\left(1-t_i\right)\log\left(1-\tilde{t}(\bm a^{(i)})\right)\right].
\end{split}
\end{equation}
Training is performed with the  gradient descent algorithm minimizing $\hatcross{p_{\mathcal{H}|\bm a}}{g}$.

In the following we show that a \ac{ce} based \ac{nn} is equivalent, in probability and for perfect training, to the \ac{np} solution. First, we prove that the output of the \ac{nn} can be interpreted as the class conditional probability. 
\begin{theorem}
Let $\gy \in [0,1]$ be the output of a \ac{nn} obtained with perfect training, i.e., with infinite number of training points, layers and neurons. Let the training be performed with the \ac{ce} metric. Then
\begin{equation}
	\gy = p_{\mathcal H|\bm a}(\mathcal{H}_0|\bm a)	
\end{equation}
almost surely.
\end{theorem}
\begin{proof}
See Appendix A.
\end{proof}

Note that this approach does not require the knowledge of the statistics of $\bm{a}$ under the two hypotheses, while, instead, it requires a large enough set of training points to converge.
However, at convergence, the \ac{nn} achieves the same performance of the NP approach. This holds since \eqref{eq:decNN} is equivalent to \eqref{eq:thrOpt}, i.e., they provide the same \ac{roc}, with
\begin{align}
	\label{eq:relation}
	\tilde{t} &= p_{\mathcal{H}|\bm a}(\mathcal{H}_0|\bm a), \quad 
	\theta = \frac{1-\lambda}{\lambda} \frac{p_\mathcal{H}(\mathcal{H}_0)}{p_{\mathcal{H}}(\mathcal{H}_1)},	
\end{align}  
where \eqref{eq:relation} is a direct consequence of the Bayes rule applied as follows 
\begin{equation}
	p_{\mathcal{H}|\bm a}(\mathcal{H}_0| \bm a) = \left[1+  \frac{p_{\mathcal{H}}(\mathcal{H}_1)}{p_{\mathcal{H}}(\mathcal{H}_0)} 2^{\mathcal{L}(\bm a)}\right]^{-1}.	
\end{equation}

% \subsection{Support Vector Machine}\label{sec:svm}
% A \ac{svm} \cite{Bishop2006} is a supervised learning model that can be used for classification and regression. We focus here on binary classification, i.e., we define the identification function as
% \begin{equation}
%   t_i =
%   \begin{cases}
%   -1 \quad \text{if} \quad \bm{y} \in \mathcal{A}_0\\
%   1 \quad \text{if} \quad \bm{y} \in \mathcal{A}_1.
%   \end{cases}
% \end{equation}
% Given the input vector $\bm{y}^{(0)} \in \mathbb{R}^N$ the \ac{svm} returns $\hat{t} = 1$ if $\bm{y}^{(0)}$ belongs to class 0 whereas $\hat{t}=-1$ if $\bm{y}^{(0)}$ belongs to class 1. It comprises the function $\tilde{t}: \mathbb{R}^N \to \mathbb{R}$ defined by
% \begin{equation}
% \label{eq:svm}
% \tilde{t} = \mathbf{w}^T \phi (\mathbf{a}^{(i)}) + b,
% \end{equation}
% where $\phi: \mathbb{R}^N \to \mathbb{R}^K$ is a feature-space transformation function, $\mathbf{w} \in \mathbb{R}^K$ is the weight vector and $b$ is a bias parameter, and the decision function is
% \begin{equation}
% \label{eq:cases}
% \hat{t} = 
% \begin{cases}
% +1 \quad \tilde{t}  \geq \gamma^* \\
% -1 \quad \tilde{t}  < \gamma^*,
% \end{cases}		
% \end{equation} 
% where $\gamma^*$ is a fixed threshold and controls \ac{fa} and \ac{md} probabilities. Note that in the classical \ac{svm} formulation we have $\gamma^* = 0$.

% While the feature-space transformation function is typically fixed, the vector $\mathbf{w}$ must be properly chosen to perform the desired classification

\section{Network Planning}\label{sec:bsPos}

As the attenuation depends on the position of the \acp{bs} and on the surrounding environment, the performance of the authentication system depends on the number of \acp{bs} and on their location. In this Section, we derive an approach to optimally locate \acp{bs} (\emph{network planning}) so that the authentication system attains the best performance.

For  \acp{bs} positioning we consider as performance metric a suitable trade-off between \ac{fa} and \ac{md} probabilities. In particular, the \ac{roc} curve associates the $P_{\rm MD}$ with the corresponding $P_{\rm FA}$, for all possible values of thresholds $\lambda$. However, as we aim at using a single performance measure without setting a priori  $P_{\rm FA}$, we  resort to the \ac{roc} \ac{auc} \cite{hanley-82}, defined as 
\begin{equation}
    C(\{\bm{x}_{\rm AP}^{(n)}\})  = \int_{0}^{1} P_{\rm MD}\left(P_{\rm FA}\right) d P_{\rm FA},
    \label{defAUC}
\end{equation}
where $P_{\rm MD}\left(P_{\rm FA}\right)$ is the $P_{\rm MD}$ value as a function of $P_{\rm FA}$. In (\ref{defAUC}) we have highlighted the dependency of the AUC on the \ac{bs} positions. Note that  $C(\{\bm{x}_{\rm AP}^{(n)}\})  $ is  the integral of the \ac{roc} function. Therefore, the \ac{bs} position optimization aims at minimizing the \ac{auc}, i.e.
\begin{equation}
    {\rm argmin}_{\{\bm{x}_{\rm AP}^{(n)}\}} C(\{\bm{x}_{\rm AP}^{(n)}\}). 
    \label{defoptpos}
\end{equation}
Note that minimizing the \ac{auc} is equivalent to minimizing the average $P_{\rm MD}$ under the assumption of a  uniformly distributed $P_{\rm FA}$.
%\footnote{Notice that traditionally the \ac{auc} is a metric that needs to be maximized \cite{hanley-82}. This is due to the fact that the curve of the system performance computed as in \cite{hanley-82} and \cite{Kennedy-11} is given by the true positive rate vs. the false negative rate value, which is optimal when the true positive rate value is maximized for each false negative rate value. Since we consider as system performance metrics for the authentication system the $P_{\rm MD}$ and the $P_{\rm FA}$ we instead need to minimize the \ac{auc}.}
In practice, in order to compute the \ac{auc} we must run the \ac{nn} over the training set  multiple times, with different thresholds and find the corresponding \ac{roc} curve, before performing its integral by numerical methods.

We hence propose to exploit the training process of the \ac{nn} and use the \ac{ce}, readily provided at the end of training, as a proxy of the system performance, avoiding an explicit estimation of  \ac{roc}  \ac{auc}. This is also motivated by Theorem 1, as the lower \ac{ce}, the more a \ac{nn} approaches \ac{np}, which is the optimal solution. 
Recalling that the training minimization problem is non-convex, the same training set can lead to \acp{nn} with different classification performance and hence different \acp{auc}. However, we select the one minimizing the   \ac{ce}, which is expected to have the minimum \ac{auc}.

\subsection{Particle Swarm Optimization}

In order to solve the network planning problem (\ref{defoptpos}) we resort to the \ac{pso} method \cite{Kennedy-11}, which is an iterative algorithm performing the simultaneous optimization of different points. This is similar to the multi-start solution for non-convex optimization, where local minima are avoided by selecting among different descent paths the one providing the minimum solution. 

The \ac{pso} method is briefly recalled here. \ac{pso} is an iterative optimization algorithm based on social behavior of animals, e.g., birds flocking and fish schools. Consider $P$ particles, where  particle $p=1, \ldots P$, is described by a vector of \acp{bs} positions $\bm{x}_p = [\bm{x}_{\rm AP}^{(1)}(p),\ldots, \bm{x}_{\rm AP}^{(N_{\rm ap})}(p)]$, and by its velocity $\bm{v}_p$.  Each particle is a candidate solution of the optimization problem. Starting from particles at random positions and velocities, at each iteration both  positions  and  velocities   are updated. Two optimal values are defined in each iteration: the global optimum found so far in the entire particle population, and a local optimum for each particle, i.e., the optimal value found by the individual $p$ up to the current iteration. We define as $\bm{o}_{\rm G}$ the position of the the global optimal values and as $\bm{o}_p$ the position of the optimal value found by particle $p$ at the current iteration. The optimal values are those minimizing the selected objective function.

The position and velocity of the particles are updated at iteration $\ell$ as \cite{Kennedy-11}
   \begin{equation}\label{eq: v up}
\begin{split}
  \bm{v}_p(\ell) = \omega \bm{v}_p(\ell-1)+\phi_1(\ell)(\bm{o}_p(\ell-1)-\\
  -\bm{x}_p(\ell-1))+\phi_2(\ell)(\bm{o}_{\rm G}(\ell-1)-\bm{x}_p(\ell-1));
  \end{split}
  \end{equation}
  \begin{equation}\label{eq: p up}
  \bm{x}_p(\ell) = \bm{x}_p(\ell-1) + \bm{v}_p(\ell),
 \end{equation}
where $\omega$ is the inertia coefficient, and $\phi_1$ ($\phi_2$) is a random variable uniformly distributed in $[0,c_1]$ ($[0,c_2]$), where $c_1$ ($c_2$) is named the {\it acceleration constant}. The inertia coefficients and acceleration constants are  parameters to be properly chosen. 

\subsection{PSO-Based Network Planning}

As we have seen the \ac{roc} \ac{auc} well describes the overall behaviour of the \ac{roc} and is hence widely recognized as a valid synthetic metric for hypothesis testing. On the other hand, \ac{auc} computation  is complicated by the need of performing extensive testing, while the \ac{ce} is immediately provided after the \ac{nn} training process. 

In particular, the testing needed to compute \ac{auc} has an additional complexity (with respect to training that must be performed anyway), of 
\begin{equation}
    \mathcal{C}_{\rm test} = P \left( \mathcal{C}_{\rm out}+\mathcal{C}_{\rm ROC}+\mathcal{C}_{\rm AUC}\right),
\end{equation}
where  $\mathcal{C}_{\rm out}$ denotes the complexity associated to running the \ac{nn} on the test points, $\mathcal{C}_{\rm ROC}$ denotes the complexity of building the \ac{roc} function, and $\mathcal{C}_{\rm AUC}$ denotes the complexity of integrating the \ac{roc}. The \ac{nn} running cost $\mathcal{C}_{\rm out}$ is given by the total number of multiplications and additions needed to compute the output value $y^{(L-1)}$ for all testing vectors, i.e.,
\begin{equation}
    \mathcal{C}_{\rm out} = \left(2N_{\rm AP}N_{\rm h}+2N_{\rm h}^2N_{\rm L} + 2N_{\rm h}\right)\tau ,
\end{equation}
where $N_{\rm h}$ is the number of neurons in the hidden layer, $N_{\rm L}$ is the number of hidden layers, and $\tau$ is the size of the testing set.
The computation of the \ac{roc} curve requires the estimation of the $P_{\rm FA}$ and $P_{\rm MD}$ values for each threshold value $\lambda$, whereas the computation of the \ac{auc} requires the numerical integration of the \ac{roc} curve over $P_{\rm FA}$ values.

The proposed \ac{pso}-based network planning algorithm is reported in Algorithm 1. We denote as $\mathcal{B}$ the optimization metric and we initialize $P$ particles with random positions for each of the $N_{\rm AP}$ \acp{bs} in each particle. For each particle we train the \ac{nn} and compute $\mathcal{B}_p^{(0)}$. The global optimum value $\mathcal{B}_g$ is set to the minimum among all $\mathcal{B}_p^{(0)}$ values. Then, positions and velocities of the particles are updated via (\ref{eq: v up}) and (\ref{eq: p up}), and both the local and global optima are updated according to the obtained values at the current iteration. The algorithm stops when the global optimum converges.

 \begin{algorithm}[b!]

\small

  \KwData{ number of particles $P$, $N_{\rm AP}$}
  \KwResult{optimal position }
  Initialize particles\;
  train the \ac{nn} algorithm for each particle\;
  $\mathcal{B}_p^{(0)}$, $p=1,\ldots,N_p$\;
  $\mathcal{B}_g=\underset{p=1,\ldots,N_p}{\min} \, \mathcal{B}_p^{(0)}$\;
  $i = 0$\;

  \Repeat{convergence of $\mathcal{B}_g$}{
         $i = i + 1$\;
         \For{$p=1,\ldots,P$}{
         update velocity and position vector of particle via (\ref{eq: v up}) and (\ref{eq: p up})\;
                  train the \ac{nn} for each particle $\to$ $\mathcal{B}_p^{(i)}$\;
                  \If{$\mathcal{B}_p^{(i)} < \mathcal{B}_g$}{ $\mathcal{B}_g = \mathcal{B}_p^{(i)}$ \;}
         }
      
      }
      
\caption{Proposed \ac{ce}-based APs positioning algorithm.}
 \end{algorithm}      

%  \begin{algorithm}[b!]

% \small

%   \KwData{ number of particles $P$, $N_{\rm AP}$}
%   \KwResult{optimal position }
%   Initialize particles\;      
%       train the \ac{nn} algorithm for each particle $\to$ $K_p^{(0)}$, $p=1,\ldots,N_p$\;
%   $K_g=\underset{p=1,\ldots,N_p}{min} \, K_p^{(0)}$\;
%       $i = 0$\;
%       \Repeat{convergence of $K_g$}{
%          $i = i + 1$\;
%          \For{$p=1,\ldots,P$}{
%          update velocity and position vector of particle via (\ref{eq: v up}) and (\ref{eq: p up})\;
%                   train the \ac{nn} for each particle \;
%                   test the \ac{nn} $\to$
%                   $K_p^{(i)}$\;
%                   \If{$K_p^{(i)} < K_g$}{ $K_g = K_p^{(i)}$ \;}
%          }
      
%       }
    
% \caption{Proposed \ac{auc}-based APs positioning algorithm.}
%  \end{algorithm}

Notice that, as the optimization problem is non-convex,  \ac{pso} is similar to a multi-start optimization with $P$ different starting points, which is a standard method used to avoid local minima. As  $P$ increases, the probability of finding only   a local solution is reduced.

\section{Numerical Results}\label{sec: nr}

The considered scenario is depicted in Fig. \ref{fig:5bs}. 
We consider a region  represented by a square with side length $525$~m, where four buildings with side length $255$~m are located at the map corners, and separated by a road with width $15$~m. The \ac{roi} $\mathcal{A}_0$ is located inside the lower-left building, delimited by the dash-dotted line. Roads are considered as \ac{los} paths, whereas transmissions from \acp{ue} located in any other map position are in non-\ac{los} conditions. $N_{\rm AP}=5$ \acp{bs} are deployed, one for each street and one at the map center, to collect attenuation values. Each \ac{bs} sees the \ac{los} path of the street it is located in. For each \ac{bs} we generate a shadowing map with standard deviation $\sigma_s = 8$ dB and decorrelation distance $d_c = 75$~m. The \ac{ue} transmits with average unitary power at the frequency $f_0 = 2.12$~GHz.

 \begin{figure}[h]
     \centering
     \includegraphics[width=0.7\columnwidth]{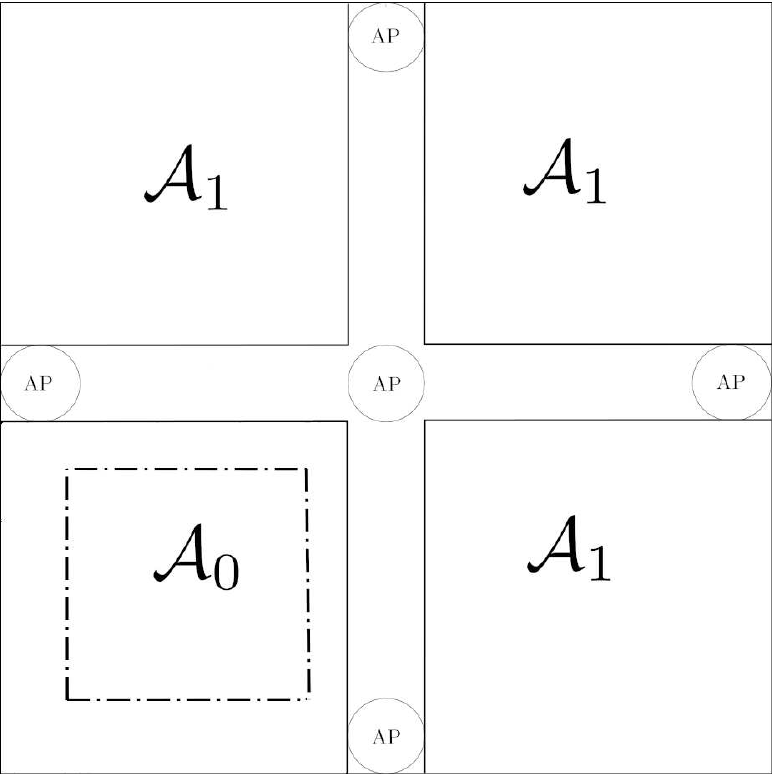}
     \caption{\ac{irlv} deployment scenario. $N_{\rm AP}=5$ \acp{bs} are located in the streets separating 4 buildings. The \ac{roi} is located inside the lower-left building, delimited by the dash-dotted line.}
     \label{fig:5bs}
 \end{figure}
 
Results are averaged over different shadowing realizations. In particular, for each \ac{fa} probability value we compute the average \ac{md} probability over different shadowing maps.
%Fig. \ref{fig:trueMap} shows a realization of the path loss and shadowing map for an \ac{bs} located in the center of the region $\mathcal{A}$. We see the two orthogonal \ac{los} paths traversing the center of $\mathcal{A}$ having lower attenuation values. The \ac{roi} $\mathcal A_0$ is delimited by the red line.

% \begin{figure}[t]
%     \centering
%     \includegraphics[width=1\columnwidth]{surfColorato.png}
%     \caption{Example of a realization of the attenuation map in the non-\ac{los} scenario considering only the shadowing effects.}
%     \label{fig:trueMap}
% \end{figure}
\subsection{In-region Location Verification Results}

Fig. \ref{fig:n_neur} shows the average (over shadowing realizations) $P_{\rm MD}$ vs. $P_{\rm FA}$ of the proposed \ac{nn} \ac{irlv} system,  with different numbers of neurons in the hidden layer $N_h$. Results have been obtained for a \ac{nn} with $N_L=3$ layers and with a training set of size $S= 10^5$. We notice that, as the number of neurons at the hidden layer increases, the average \ac{fa} probability decreases. When the number of neurons $N_h$ is higher than $8$ however we notice that results converge, meaning that increasing the network size does not lead to a performance improvement. Therefore, in the following we set $N_h=8$.

Fig. \ref{fig:n_train} shows the average (over shadowing realizations) $P_{\rm MD}$ vs. $P_{\rm FA}$ of the proposed \ac{nn} \ac{irlv} system  trained with different numbers of training points $S$. Results have been obtained for a \ac{nn} with  $L=3$ layers and $N_h=8$ neurons in the hidden layer. We see that the \ac{auc} decreases when increasing the number of training points and that, starting from $S=10^5$, the \ac{roc} does not significantly improve. This is due to the fact that, for the selected \ac{nn} architecture, training reaches convergence and hence adding further  training points does not improve the \ac{nn} performance. 

\begin{figure}[t]
    \centering
    \includegraphics[width=1\columnwidth]{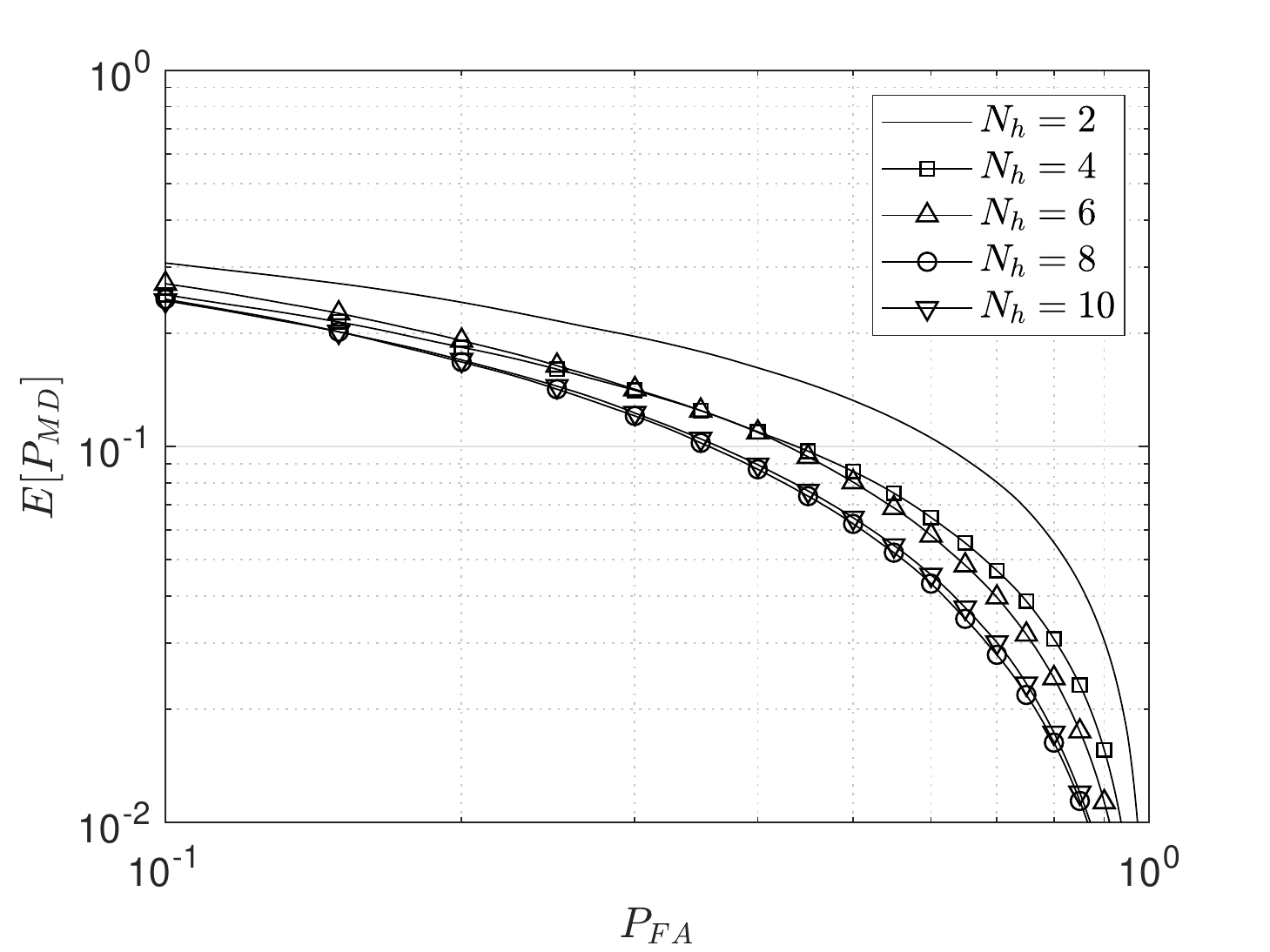}
    \caption{\ac{roc} of the \ac{nn} \ac{irlv} system with different numbers of neurons in the hidden layer $N_h$.}
    \label{fig:n_neur}
\end{figure}

\begin{figure}[t]
    \centering
    \includegraphics[width=1\columnwidth]{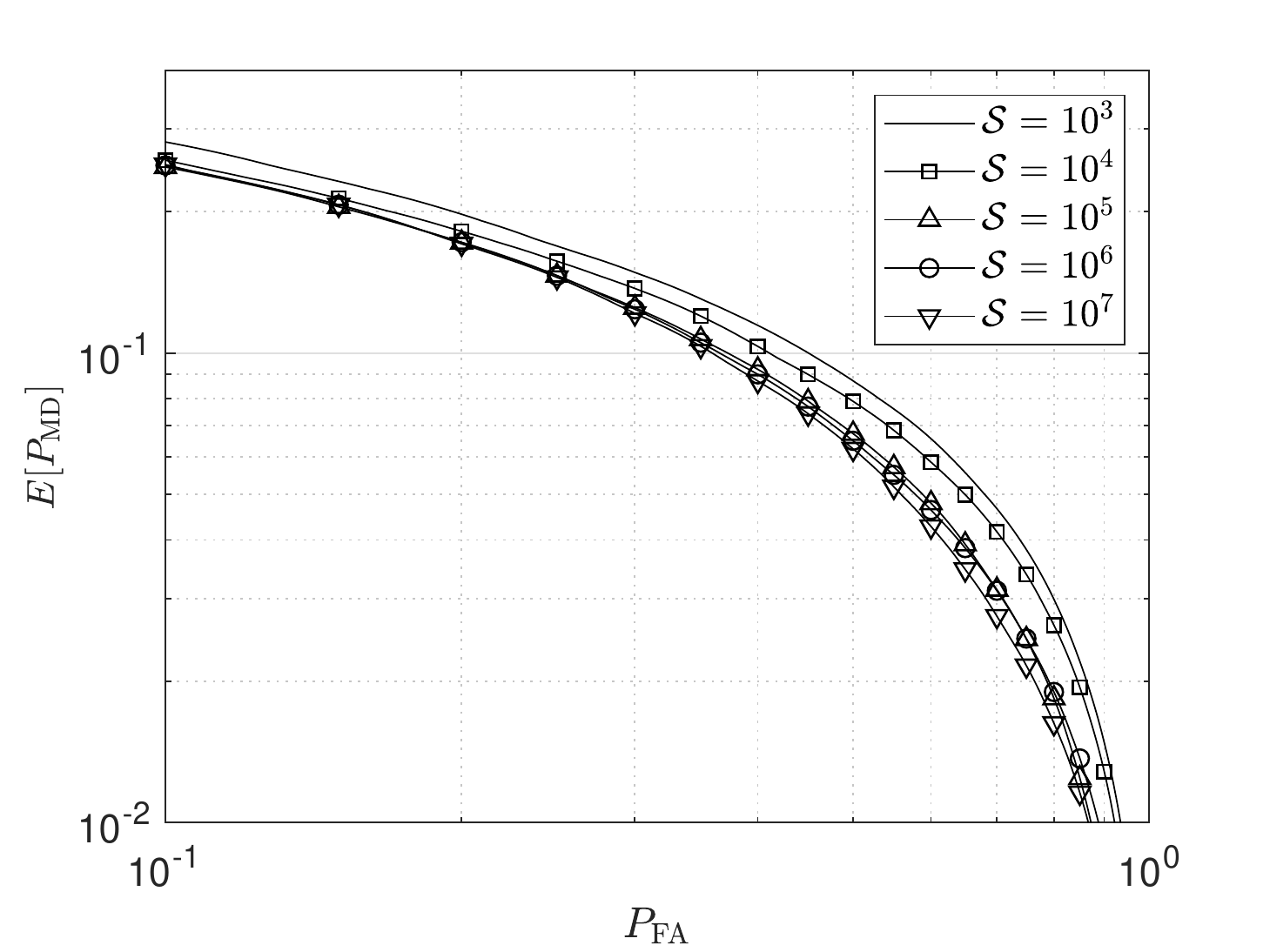}
    \caption{\ac{roc} of the \ac{nn} \ac{irlv} system trained with $S$ training points.}
    \label{fig:n_train}
\end{figure}

Fig. \ref{fig:NP_comp} shows part of the \ac{roc} obtained with the \ac{np} test and with the \ac{nn}, using  the model of Section \ref{sec:los}. In particular, we consider an overall circular region $\mathcal{A}$ with radius $R_{\rm out} = 40$~m, a square authentic region $\mathcal{A}_0$ of $L = H= 25$~m located inside $\mathcal{A}$, with upper left corner at a distance of $R_{\rm min} = 4$~m from the center of $\mathcal{A}$. We also report the results of \ac{np} theorem, that can be computed in close form for this simple scenario.  We see that, even with a small number of neurons, in this simple problem, the \ac{nn} achieves the same \ac{roc} of the \ac{np} test, thus confirming  Theorem 1.

 \begin{figure}[h]
     \centering
     \includegraphics[width=0.9\columnwidth]{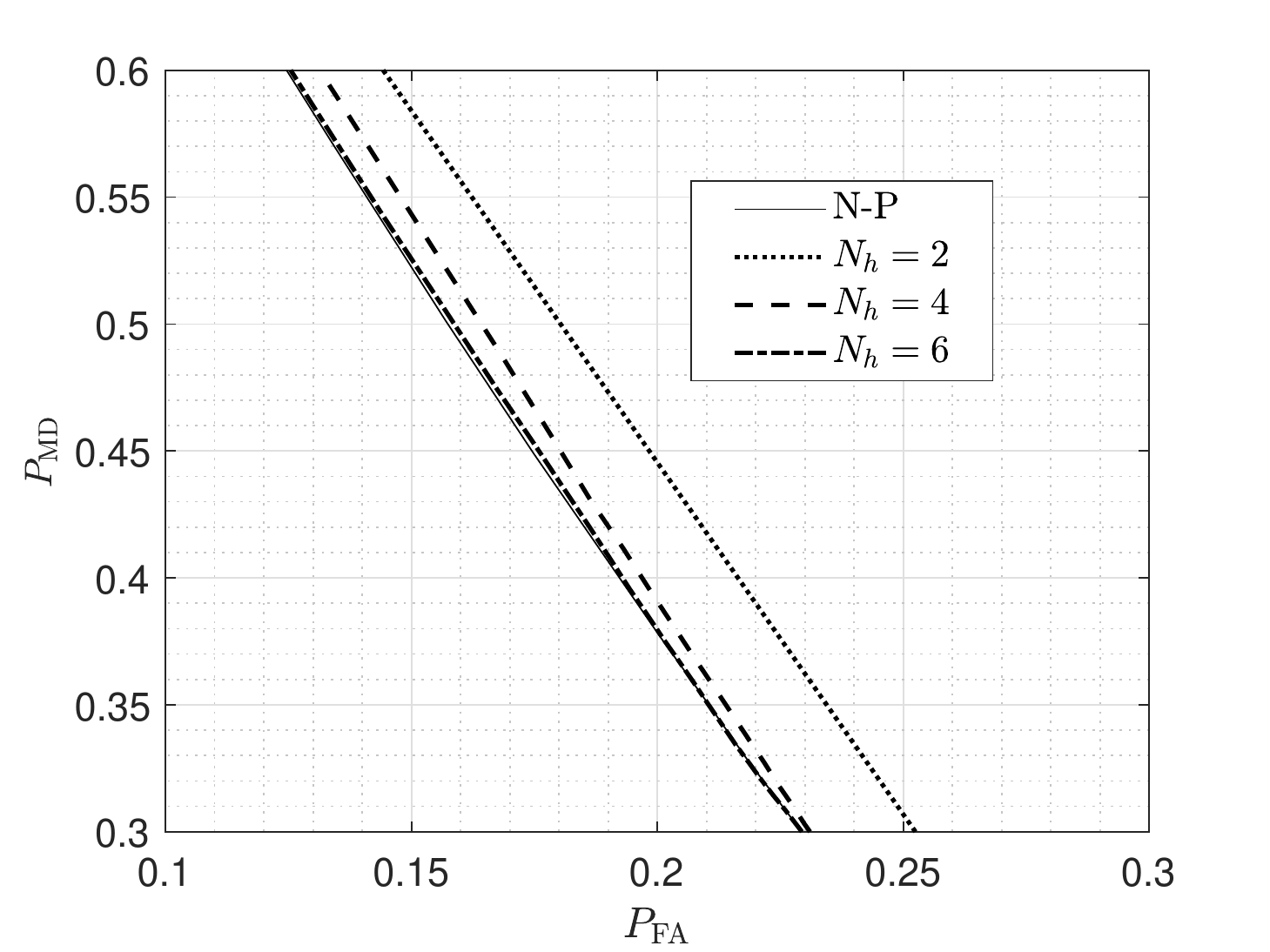}
     \caption{\ac{roc} of the \ac{np} test and the the proposed \ac{ml} test,  with different number of neurons in the hidden layer of the \ac{nn}.}
     \label{fig:NP_comp}
 \end{figure}

\subsection{Network Planning Performance}

We now consider the network planning problem, using the proposed Algorithm 1. We consider a \ac{pso} with $P=6$ particles, each composed by a set of $N_{\rm AP}=5$ \acp{bs} initialized with random positions. There exists a variety of implementations of the \ac{pso}, but the most general case for the parameter initialization is given by \cite{clerc2002}, where it is suggested to set $\omega=0.7298$, and $c_1=c_2=1.4961$. Results are averaged over different shadowing realizations. The used \acp{nn} are implemented with $L=3$ layers and $N_h=8$ neurons at the hidden layer. In order to validate the use of the \ac{ce} as a proxy for the \ac{auc}, we compare the performance of Algorithm 1 using either the CE or the AUC as objective function  $\mathcal{B}$. For a performance assessment we show  the \ac{auc} obtained with the various optimization approaches.

Fig. \ref{fig:CEvsAUC} shows the average \ac{auc}  vs. the number of \ac{pso} iterations for the two implementations of Algorithm 1, and for different training set size $S$.
We recall that the result of Theorem 1 holds asymptotically and with perfect training (in terms of training set, number of layers and neurons), which are conditions that our experiments do not satisfy. Two observations are in place here. 

First, the training set size should be sufficiently large so that \ac{ce} becomes a proxy of \ac{auc}: in fact we notice that for $S=10^3$, the  mean \ac{auc}  for \ac{pso} with \ac{ce} objective function increases with the  number of iterations. This is due to the fact that convergence has not been reached because of the limited size of the training set, and hence the obtained \ac{ce} value is not representative of the \ac{irlv} system performance. This is solved by augmenting the training set size, as we see when  $S=10^4$ or $10^5$. 

Second,  the \ac{auc}-based algorithm converges earlier than the \ac{ce}-based algorithm. As stated earlier, the \ac{nn} converges to the optimal \ac{np} solution only asymptotically. Therefore the obtained \ac{ce} is an approximation of the selected objective function, whose optimization does not necessarily entail the minimization of the \ac{auc}. Although requiring a larger number of iterations, the \ac{ce}-based solution is convenient as it does not require the explicit estimation of the \ac{roc} curve.   

%Fig. \ref{fig:CEvsAUC} shows the average \ac{auc} value vs. the number of \ac{pso} iterations. We see that the \ac{ce}-based solution reaches approximately the same  \ac{auc} value of the \ac{auc}-based solution as the number of iterations increases. Recalling that \eqref{eq:final} is obtained as asymptotic result, the fact that the \ac{ce} based solution requires a larger number of iteration to reach the minimum is due to the fact that results obtained by simulations consider finite number of neurons and finite training set. However we notice that also in the finite case the \ac{ce} minimization provides a good approximation of the classification performance of the network.

% Fig. \ref{fig:cdf} shows the \ac{cdf} of the number of  iterations using \ac{ce} as objective function in Algorithm 1. We note that in half of the cases only four iterations are performed using \ac{ce} and then the rest of the iterations are using the \ac{auc} as a target function. Since from Fig. \ref{fig:CEvsAUC} we note that in five  iterations Algorithm 1 is on average already quite close to convergence we can conclude that the number of iterations in the \ac{auc} stage is very small in most cases. This does not occur for the pure \ac{auc} solution, where a higher number of \ac{auc} iterations is needed to achieve convergence (see  Fig. \ref{fig:CEvsAUC}).
%\balance
\begin{figure} 
    \centering
    \includegraphics[width=1\columnwidth]{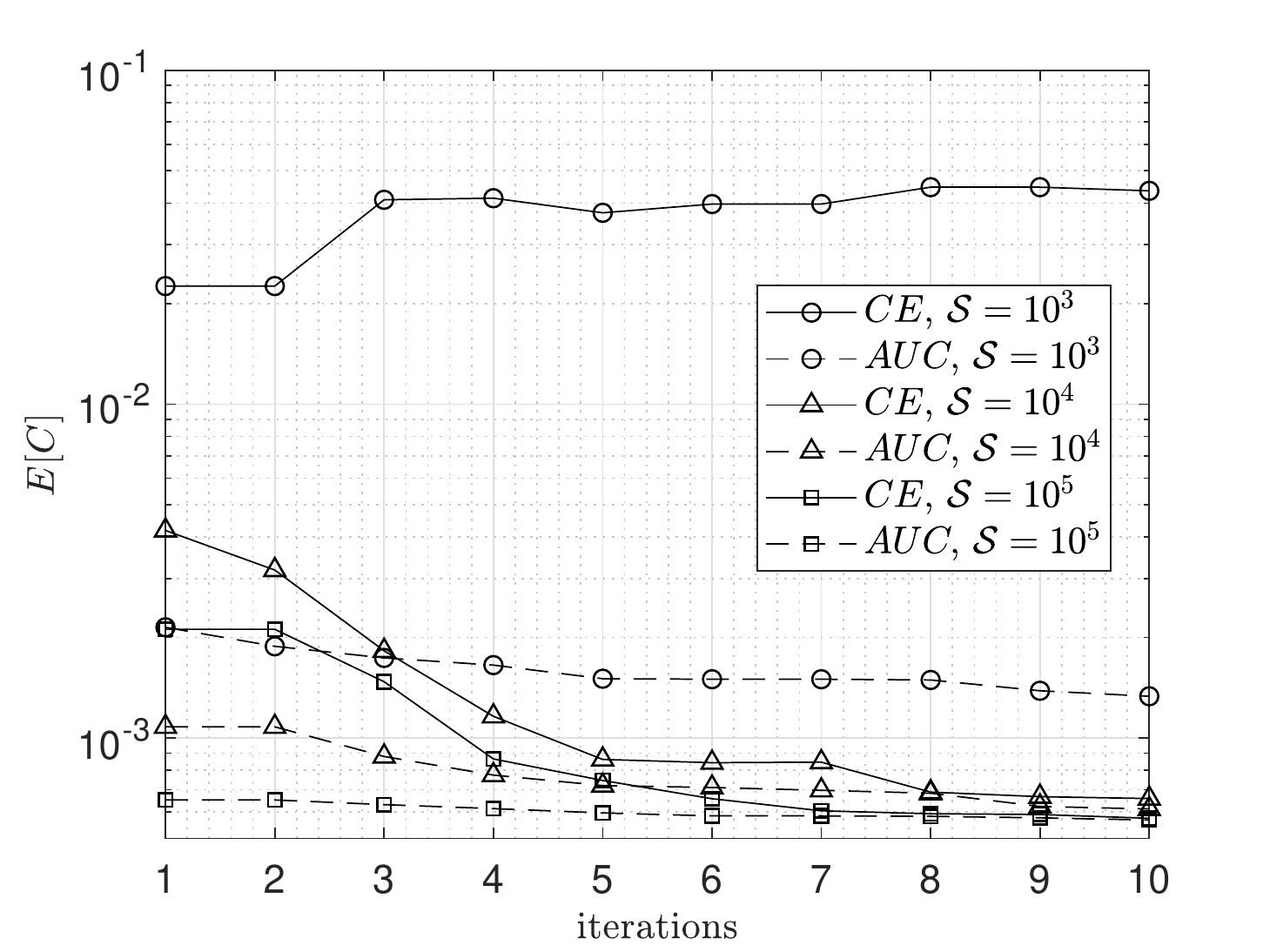}
    \caption{Mean \ac{auc} vs. the number of \ac{pso} iterations for Algorithm 1, and two \ac{pso} algorithms using only the  \ac{auc} and the \ac{ce} as objective functions. }
    \label{fig:CEvsAUC}
\end{figure}

% \begin{figure} 
%     \centering
%     \includegraphics[width=0.9\columnwidth]{cdf_bar.eps}
%     \caption{CDF of the number of iterations of the first stage of Algorithm 1. }
%     \label{fig:cdf}
% \end{figure}
 
\section{Conclusions}
\label{sec:conc}
%\balance

In this paper we formulated the \ac{irlv} problem as an hypothesis testing problem and proposed a \ac{ml} solution. We proved that the \ac{nn} implementation achieves the same performance of the optimal \ac{np} test with \ac{ce} as training objective function, and verified numerically this claim for a simple scenario. We also assessed the effects of the training set size over the \ac{roc} for a more realistic scenario. We then proposed a \ac{pso} algorithm for   optimal \acp{bs} positioning, showing that there is a minimum training set size which allows to use the \ac{ce} as a proxy of the \ac{auc}.
% We then showed by numerical results the effectiveness of the proposed solution in terms of \ac{auc} of the \ac{roc}.

\begin{appendices}
\section{Proof of Theorem 1}
Consider
\begin{equation}
\begin{split}
	\hatcross{ p_{\mathcal{H}|\bm a}}{g} =&  - \left[ \frac{n_0}{S} \frac{1}{n_0} \sum_{\bm a \in \mathcal{T}_0} \log(\gy) \right. \\
		&\left. + \frac{n_1}{S} \frac{1}{n_1} \sum_{\bm a \in \mathcal{T}_1} \log(1-\gy) \right],	
\end{split}
\end{equation}
where $\mathcal{T}_k = \{\bm a_i \in \mathcal{T} : t_i = k\}$, $|\mathcal{T}_k|=n_k$, $k \in \{0,1\}$.

\balance

Let $S \to \infty$. By the strong law of large numbers% \cite{etemadi1981elementary}
\begin{equation}
\label{eq:as}
\begin{split}
	\lim_{S \to \infty}&	\hatcross{ p_{\mathcal{H}|\bm a}}{g} \simeq - \left[ p_{\mathcal H}(\mathcal{H}_0) \int_{\mathcal{Y}} p_{\bm a|\mathcal{H}}(\bm a|\mathcal{H}_0) \log (\gy) d\bm a \right. \\
	& \left. + p_{\mathcal{H}}(\mathcal{H}_1) \int_{\mathcal{Y}} p_{\bm a|\mathcal{H}}(\bm a|\mathcal{H}_1) \log (1-\gy) d\bm a \right],
\end{split}
\end{equation}
where equality $\simeq$ holds in probability, i.e., \textit{almost surely}, as per the strong law of large numbers.
Rearranging terms with the Bayes rule we get
\begin{equation}
\label{eq:dim1}
\begin{split}
\lim_{S \to \infty}	\hatcross{p_{\mathcal{H}|\bm a}}{g} \simeq - \left\{ \int_{\mathcal{Y}} \bigl[ p_{\mathcal{H}|\bm a}(\mathcal{H}_0|\bm a) \log \gy + \right. \bigr.\\ 
	\bigl. \bigl.(1-p_{\mathcal{H}|\bm a}(\mathcal{H}_0|\bm a)) \log(1-\gy)\bigr] p(\bm a)   d\bm a \biggr\}. 		
\end{split}
\end{equation}
By definition of expected value we can rewrite \eqref{eq:dim1} as
\begin{equation}
		\label{eq:dim2}
		\begin{split}
	\lim_{S \to \infty}	\hatcross{p_{\mathcal{H}|\bm a}}{g} \simeq - \mathbb{E}_{\bm a} \left[ p_{\mathcal{H}|\bm a}(\mathcal{H}_0|\bm a) \log \gy + \right.\\ 
			\left. (1-p_{\mathcal{H}|\bm a}(\mathcal{H}_0|\bm a)) \log(1-\gy)\right].  		
		\end{split}
\end{equation}
We introduce the Bernoulli random variable $\xi$ with alphabet $ \{0,1\}$ and \ac{pdf} 
\begin{equation}
\label{eq:q}
p_\xi(0) = \gy, \quad
p_\xi(1) = 1- \gy.
\end{equation}
Note that $p_\xi$ is a valid \ac{pdf} since it sums to 1 and $\gy \in [0,1]$ by hypothesis.
Recall now that the cross entropy between two discrete \acp{pdf} $p_{W_1}(w)$ and $p_{W_2}(w)$ having the same alphabet $\mathcal{W}$ is defined as
\begin{equation}
\label{eq:defCross}
	\cross{p_{W_1}}{W_2} = - \sum_{w \in \mathcal{W}} p_{W_1}(w) \log p_{W_2}(w),
\end{equation}
and an equivalent definition is
\begin{equation}
\label{eq:alternativeDef}
		\cross{p_{W_1}}{W_2} = H (p_{W_1}) + D(p_{W_1}||p_{W_2}),
\end{equation}
where $D(\cdot||\cdot)$ is the Kullback-Leibler (K-L) divergence and $H(\cdot)$ is the entropy function.
Then, form \eqref{eq:dim2} and \eqref{eq:defCross}, we have 
\begin{equation}
		\lim_{S \to \infty}	\hatcross{p_{\mathcal{H}|\bm a}}{g} \simeq \mathbb{E}_{\bm a} \left[ \cross{p_{\mathcal{H}|\bm a}}{\xi} \right],
\end{equation}
which from \eqref{eq:alternativeDef} yields
\begin{equation}
\label{eq:final}
		\lim_{S \to \infty}	\hatcross{p_{\mathcal{H}|\bm a}}{g} \simeq 
		\E{\bm a}{H ( p_{\mathcal{H}|\bm a}) + D( p_{\mathcal{H}|\bm a}|| p_\xi)}
\end{equation}

Recall that \ac{nn} training is performed by minimizing the left hand side of \eqref{eq:final} with respect to \ac{nn} parameters. In the right hand side of \eqref{eq:final} the only quantity depending on \ac{nn} parameters, through $\gy$ in \eqref{eq:q}, is $D( p_{\mathcal{H}|\bm a}||p_\xi)$. Then, with a infinite number of neurons (i.e., with the possibility of choosing any \ac{pdf} $p_\xi$, the minimum of the K-L divergence is  attained for $D(p_{\mathcal{H}|\bm a}||p_\xi)=0$, that is when $p_\xi(i) = p_{\mathcal{H}|\bm a}(i|\bm{a})$.
\end{appendices}

\bibliographystyle{plain}
\bibliography{bibliography2}

\end{document}